\documentclass[conference]{IEEEtran}
\IEEEoverridecommandlockouts
\usepackage{verbatim}
\usepackage{epsfig}
\usepackage{amsmath}
\usepackage{lipsum}
\usepackage{amsthm}
\usepackage{amssymb}
\usepackage{psfrag}
\usepackage{soul}
\usepackage[T1]{fontenc}
\usepackage[ruled]{algorithm2e}
\usepackage{mathtools, nccmath}
\usepackage{cuted}
\usepackage{multirow}
\usepackage[table,xcdraw]{xcolor}
\usepackage{etoolbox}
\AfterEndEnvironment{strip}{\leavevmode}
\usepackage{url}
\usepackage{dsfont}
\usepackage{mathtools}
\usepackage{bbm}

\usepackage[usenames,dvipsnames]{pstricks}
\usepackage{etoolbox}
\usepackage{graphicx}
\graphicspath{{Images/}}
\ifCLASSOPTIONcompsoc
\usepackage[caption=false,font=footnotesize,labelfont=sf,textfont=sf]{subfig}
\else
\usepackage[caption=false,font=footnotesize]{subfig}
\fi
\usepackage{stfloats}
\makeatletter
\patchcmd{\@makecaption}
  {\scshape}
  {}
  {}
  {}
\makeatother
\usepackage[noadjust]{cite}

\usepackage[font=footnotesize,skip=0pt]{caption}

\newtheorem{theorem}{Theorem}[]
\newtheorem{lemma}{Lemma}[]

\newtheorem{remark}{Remark}

\DeclareMathOperator*{\minimize}{minimize}

\definecolor{color1}{rgb}{0,0,0}
\usepackage{textcomp}
\usepackage{xcolor}
\def\BibTeX{{\rm B\kern-.05em{\sc i\kern-.025em b}\kern-.08em
    T\kern-.1667em\lower.7ex\hbox{E}\kern-.125emX}}
\IEEEpeerreviewmaketitle 
\begin{document}
\title{A Plug and Play Distributed Secondary Controller for Microgrids with Grid-Forming Inverters \thanks{
This research is supported by the United States Department of Energy via grant number DE-CR$0000040$.}}
\author{\IEEEauthorblockN{Vivek Khatana$^{1}$, Soham Chakraborty$^{2}$, and Murti V. Salapaka$^{3}$}
\IEEEauthorblockA{Department of Electrical and Computer Engineering, University of Minnesota, USA
\\ \{$^1$khata010, $^2$chakr138, $^3$murtis\}@umn.edu}}
\maketitle

\begin{abstract}
A distributed controller for secondary control problems in microgrids with grid-forming (GFM) inverter-based resources (IBRs) is developed. The controller is based on distributed optimization and is synthesized and implemented distributively enabling each GFM IBR to utilize decentralized measurements and the neighborhood information in the communication network. We present a convergence analysis establishing voltage regulation and reactive power sharing properties. A controller-hardware-in-the-loop experiment is conducted to evaluate the performance of the proposed controller. The experimental results corroborate the efficacy of the proposed distributed controller for secondary control.\\[1ex]
\end{abstract}

\IEEEpeerreviewmaketitle
\section{Introduction}
AC microgrids (MGs) consist of loads and inverter-based resources (IBRs) interfacing with renewable energy resources and energy storage via power electronics. Based on the nature and the expected functionality, these IBRs usually operate in a grid-following (GFL) or grid-forming (GFM) mode. IBRs operating in the GFM mode should generate and maintain stable voltages and frequency of the MG when disengaged from the main grid. MGs typically follow a hierarchical control structure by adapting ANSI/ISA-$95$-based international standard \cite{guerrero2010hierarchical}. The hierarchical control structure consists of three levels: primary, secondary, and tertiary control. The concept of droop control is utilized extensively at primary control for active and reactive power sharing \cite{chandorkar1993control}. However, for MG operation with GFM IBRs with only primary droop control, the voltage of the system deviates from its nominal values, and reactive power demand is not shared as desired. One of the main challenges at the secondary control level is the coordination of the GFM IBRs for restoring system voltages to the nominal value and sharing reactive power appropriately. 
\par Traditional centralized secondary control systems \cite{guerrero2010hierarchical, zhong2011robust, savaghebi2012secondary, micallef2014reactive} require global information aggregation and dissemination and don't scale well. Researchers have explored distributed secondary control schemes to reduce communication bandwidth requirements. Most distributed secondary control schemes include modification of the droop control with an additive corrective term (see \cite{shafiee2012distributed, shafiee2013distributed, shafiee2013robust, bidram2013secondary, dehkordi2016fully, simpson2013synchronization, simpson2015secondary, chen2015distributed, lu2016novel, lu2016distributed, ding2018distributed, zhang2024secondary} and the references therein). The distributed voltage controllers proposed in \cite{shafiee2012distributed, shafiee2013distributed, shafiee2013robust} require each GFM IBR to collect measurements of every GFM IBR in the network to produce the appropriate local control signal; leading to a huge communication bandwidth for every GFM IBR. Further, the private information of each GFM IBR is exposed to the entire network. In articles \cite{bidram2013secondary,dehkordi2016fully} communication networks with a directed spanning tree are considered. Further, to generate voltage control set-points, a leader-follower synchronization scheme is utilized that is susceptible to a single point of failure. Articles \cite{simpson2013synchronization, simpson2015secondary} developed distributed proportional and integral (PI) controllers for frequency and voltage regulation via averaging dynamics based on a balanced positive definite Laplacian matrix. Although the PI controller can be implemented in a distributed manner it requires a centralized synthesis. Articles \cite{chen2015distributed, lu2016novel} propose a centrally synthesized non-linear non-smooth distributed control signal based on a leader-follower synchronization and suffers from the same limitations as \cite{bidram2013secondary,simpson2013synchronization}. The work in article \cite{lu2016distributed} considers discrete-time updates with non-uniform communication intervals for a virtual leader-follower distributed PI controller. Article \cite{ding2018distributed} utilizes an event-triggered communication approach to reduce the utilization of the communication. However, the controller requires global information for synthesis and is unsuitable for plug-and-play operation. Article \cite{zhang2024secondary} presents an optimization problem for the design of a sparsity-promoting Laplacian matrix utilized in the centrally synthesized controllers schemes in \cite{simpson2013synchronization,simpson2015secondary}. Manipulating the virtual impedance (VI) in the control loop provides another approach for controlling reactive power sharing among GFM IBRs (see \cite{zhu2015virtual, hu2018decentralised, zhang2016distributed} and references therein). In \cite{zhu2015virtual}, the VI parameters are tuned via a genetic algorithm, introduced to minimize the system's global reactive power sharing error. Article \cite{hu2018decentralised} utilizes a VI proportional to the reactive power. The VI design in \cite{zhu2015virtual, hu2018decentralised} is performed offline in a centralized MG configuration stage. The scheme in article \cite{zhang2016distributed} considers an adaptive VI designed via a balanced Laplacian matrix.  \\ Existing secondary control schemes in the literature do not prescribe a distributed synthesis for controllers and hence do not have a plug-and-play operation. Here, we propose a novel secondary control scheme for reactive power sharing and voltage regulation among heterogeneous GFM IBRs. The major
contributions of this article are:\\
$1)$ A distributed secondary control strategy for voltage regulation and reactive power sharing is proposed based on a distributed optimization framework. Compared to the existing works in the literature the proposed scheme allows for a fully distributed synthesis and a plug-and-play operation. \\
$2)$ The proposed scheme provides reactive power sharing and voltage regulation among heterogeneous distributed GFM IBRs connected through varying line impedances.\\
$3)$ We utilize an online algorithm using only local measurements to lower the communication bandwidth requirements. Additionally, only non-linear estimates are shared between the GFM IBRs during updates, safeguarding private information such as real-time power measurement, maximum capacity of generations, and local load demands, which is beneficial to enhance the protection of information privacy against malicious cyber-attacks.\\
We evaluate the proposed distributed secondary control scheme using a controller hardware-in-the-loop (CHIL)-based real-time simulation. The laboratory test results establish the efficacy of the proposed scheme for secondary control in MGs.
\section{Brief Description of Secondary Control} \label{sec:motivation}
\begin{figure}[b]
    \centering
    \subfloat[]{\includegraphics[scale=0.26,trim={0cm 0cm 16.5cm 5.8cm},clip]{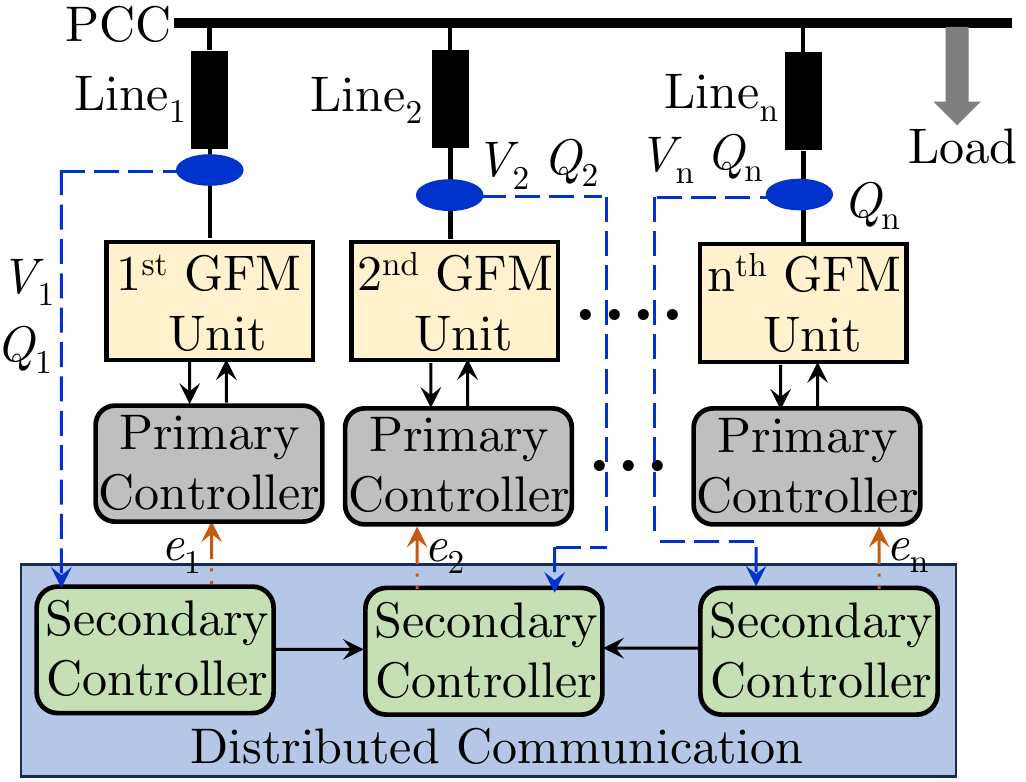}}%
    \label{fig:system}~
    \subfloat[]{\includegraphics[scale=0.27,trim={0cm 0cm 18.5cm 6cm},clip]{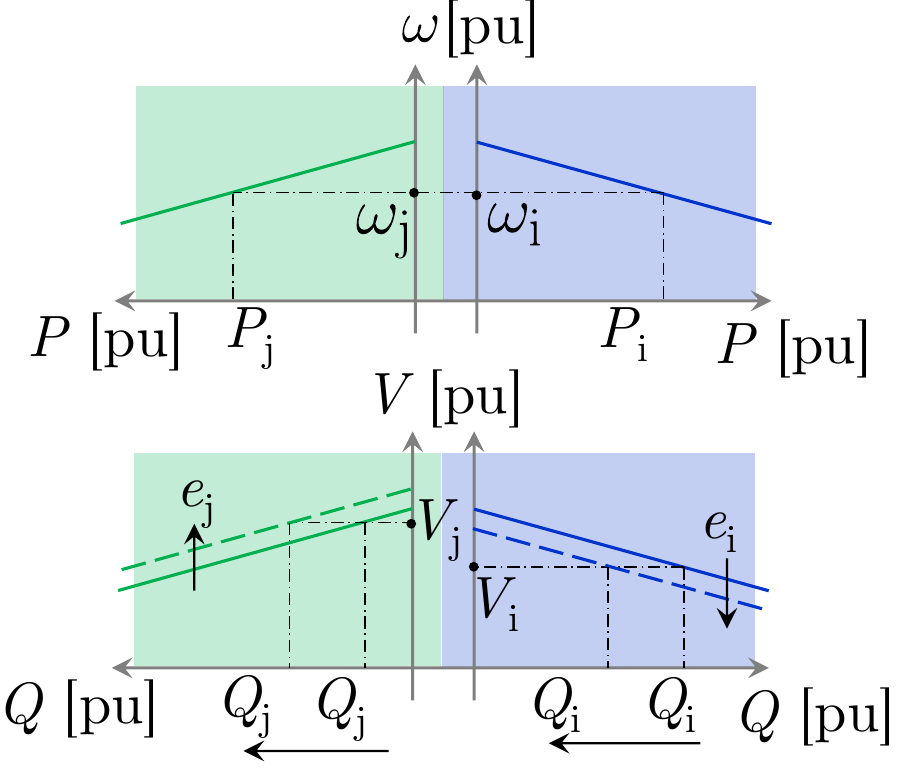}}%
    \label{fig:concept}
    \caption{(a) Distributed $\textit{secondary~control}$ architecture with $\textit{primary~control}$ for multiple GFM IBRs connected in parallel, (b) $P$$\sim$$f$, $Q$$\sim$$V$ droop-based $\textit{primary~control}$ with $\textit{secondary~control}$ for parallel GFM IBRs $i$ and $j$ with equal ratings, connected to PCC via reactive lines with $X_\mathrm{j} > X_\mathrm{i}$.}
    \label{fig:system_concept}
\end{figure}
\noindent In this section, secondary control is discussed from the scope of reactive power sharing control among multiple GFM IBRs in an islanded MG. Fig.~\ref{fig:system_concept}(a) shows a generic MG consisting of $N$ number of GFM IBRs connected to a common bus, point-of-common-coupling (PCC), via various lines. A complete description of the control typologies of the GFM IBRs is out of scope; a brief description is provided. Each GFM IBR is controlled via voltage controller-based $\textit{inner~control~loop}$ where the reference voltage signal is generated using the $P$$\sim$$f$/$Q$$\sim$$V$ droop control-based $\textit{primary~control}$. For inductive lines, the $\textit{primary~controller}$ without the $\textit{secondary~control}$ specify the $i^\mathrm{th}$ GFM IBR's frequency, $\omega_\mathrm{i}$, and voltage magnitude, $V_\mathrm{i}$ by:
\begin{align}\label{droop_law}
    \omega_\mathrm{i} = \omega^* - n_\mathrm{i}P_\mathrm{i},~V_\mathrm{i} = V^* - m_\mathrm{i}Q_\mathrm{i},
\end{align}
where $\omega^*$, $V^*$ are the nominal frequency (in rad/s) and voltage set-point (in volt) of the MG respectively, and the gains $n_\mathrm{i}$, $m_\mathrm{i}$ are the droop coefficients. $P_\mathrm{i}$ and $Q_\mathrm{i}$ are the measured average active and reactive power of the $i^\mathrm{th}$ GFM IBR, respectively. It is generated by processing the instantaneous active power, $p_\mathrm{i}$, and reactive power, $q_\mathrm{i}$, via a low-pass filters with the time constant, $\tau_\mathrm{S,i} \in \mathbb{R}_{>0}$. As a result, $P_\mathrm{i} :=  [1/(\tau_\mathrm{S,i}s+1)]p_\mathrm{i}$ and $Q_\mathrm{i} :=  [1/(\tau_\mathrm{S,i}s+1)]q_\mathrm{i}$. The large-signal stability
analysis of the MG system yields the steady-state network frequency as,
\begin{align}\label{ss_freq}
    \omega_\mathrm{ss} = \textstyle \omega^* - \left(\frac{1}{\sum_{i=1}^{n}(1/n_\mathrm{i})}\right) P_\mathrm{L},
\end{align}
where $P_\mathrm{L}$ is the total active power load in the MG. The steady-state frequency, $\omega_\mathrm{ss}$, being a global quantity of the MG system, is different from the nominal $\omega^*$, but equal in all the GFM IBRs. Thus, using~\eqref{droop_law} and~\eqref{ss_freq}, the following can be stated as,
\begin{align}
    n_1P_{1} = n_2P_{2} = \ldots = n_\mathrm{i}P_\mathrm{i} = \ldots = n_\mathrm{n}P_\mathrm{n}.
\end{align}
Therefore, the total active power demand, $P_\mathrm{L}$, is shared according to the selection of the $P$$\sim$$f$ droop coefficients. In other words, the per-unitized (pu) active power with each GFM IBR's base rating, is equal. This phenomenon can be understood using the top figure of Fig.~\ref{fig:system_concept}(b). However, the same analysis for the $Q$$\sim$$V$ droop controller is not as straightforward. It is well studied that due to the unequal line impedances, the $Q$$\sim$$V$ droop controller is unable to share reactive power demand among even identical GFM IBRs operating in parallel and as a result, 
\begin{align}\label{q_share_1}
    m_1Q_{1} \neq m_2Q_{2} \neq \ldots \neq m_\mathrm{i}Q_\mathrm{i} \neq \ldots \neq m_\mathrm{n}Q_\mathrm{n}.
\end{align}
For simplicity of exposition, a case study of two GFM IBRs with identical ratings ($i^\mathrm{th}$ and $j^\mathrm{th}$ GFM unit), but operating through reactive lines with $X_\mathrm{j} > X_\mathrm{i}$ in parallel is considered. The bottom figure of Fig.~\ref{fig:system_concept}(b) depicts the $Q$$\sim$$V$ droop law before the $\textit{secondary~control}$ action (solid lines of bottom figure of Fig.~\ref{fig:system_concept}(b)). The GFM IBRs operate at voltages $V_\mathrm{i}$ and $V_\mathrm{j}$ with reactive power injections $Q_\mathrm{i}$ and $Q_\mathrm{j}$. Since, $X_\mathrm{j} > X_\mathrm{i}$, the terminal voltages of the GFM IBRs follows $V_\mathrm{i}<V_\mathrm{j}$ and as a result, reactive power flow follows $Q_\mathrm{i}>Q_\mathrm{j}$. This is the limitation of employing only $\textit{primary~control}$ without the $\textit{secondary~control}$ compensation in parallel operation of multiple GFM IBRs. Moreover due to non-zero values of $Q_\mathrm{i}$'s and voltage drop across the line impedances, the voltage at PCC becomes less than nominal, $V^*$. Therefore, to ensure reactive power sharing among multiple GFM IBRs connected to the PCC via asymmetrical line impedances and desired voltage regulation at PCC, a dynamic compensation of $Q$$\sim$$V$ droop law is required as discussed in the next section.
\section{Distributed Reactive Power Sharing Scheme}\label{sec:formulation}
\subsection{Control/Update Rule Design} 
Consider a MG with $N$ GFM IBRs. Let the voltage and the reactive power at the $i^\mathrm{th}$ GFM IBR at any time be denoted as $V_\mathrm{i}(t)$ and $Q_\mathrm{i}(t)$ respectively. A droop law is employed for updating the voltage at the $i^\mathrm{th}$ GFM IBR any time $t+1$,
\begin{align}\label{eq:voltage_droop}
    V_\mathrm{i}(t+1) = V^* - m_\mathrm{i} Q_\mathrm{i}(t) + v_\mathrm{i}(t),
\end{align}
where $V^*$ is the voltage reference, $m_\mathrm{i}$ is the droop coefficient, and $v_\mathrm{i}(t)$ is the adjustment given at time $t$. The adjustment $v_\mathrm{i}(t)$ should either maintain equal-rated reactive power sharing between the GFM IBRs or tight regulation of the voltage at the GFM IBRs depending on the system requirement. We follow an optimization-oriented approach to design the adjustment $v_\mathrm{i}(t)$. Let the directed graph $\mathcal{G}$ represent the communication network of the GFM IBRs ($\mathcal{G}$ isn't necessarily the same as the power network connections). Let $\mathcal{G}(\{1,2,\dots,N\},\mathcal{E})$ denote a graph with the set of edges $\mathcal{E} \subset \{1,2,\dots,N\} \times \{1,2,\dots,N\}$. Let $\mathcal{N}_\mathrm{i}^{-} := \{j \ | \ (i,j)\in \mathcal{E}\}$ and $\mathcal{N}_\mathrm{i}^{+} := \{j \ | \ (j,i)\in \mathcal{E}\}$ respectively denote the set of in-neighbors and out-neighbors of node $i$ in  $\mathcal{G}$. Consider the  optimization problem at time $t$,
\begin{align}\label{eq:opt_prob}
    \minimize_{x_1, x_2, \dots, x_\mathrm{N}} & \ \textstyle \sum_{i=1}^N  \frac{1}{2}(x_\mathrm{i} - \alpha_\mathrm{i}(t))^2 + \frac{\gamma}{2} x_\mathrm{i}^2 \\
    \mbox{subject to} \ \alpha_\mathrm{i}(t):&= a_\mathrm{V_i}(V^* - V_\mathrm{i}(t)) \nonumber \\
                                        & + a_\mathrm{Q_i} m_\mathrm{i} Q_\mathrm{i}(t) \  \mbox{for all} \ i \in \{1,\dots,N\}, \label{eq:alpha_constraint} \\
                    x_\mathrm{i} &= x_\mathrm{j} \  \mbox{for all} \ i,j \in \{1,2,\dots,N\},\nonumber 
\end{align}
where $a_\mathrm{V_i}$ and $a_\mathrm{Q_i}$ are free parameters leading to different objective functions and $\gamma > 0 $ is a regularization parameter. Let $x_\mathrm{t}^* = x_1^*=\dots=x_\mathrm{N}^*$ denote the solution to the optimization problem~\eqref{eq:opt_prob}. Then the adjustment $v_\mathrm{i}(t)$ at time $t$, is given by the following control law, for all $i \in \{1,2,\dots, N\}$,
\begin{align}\label{eq:adjustmet_update}
    v_\mathrm{i}(t) := x_\mathrm{t}^* + \beta_\mathrm{Q_i} m_\mathrm{i} Q_\mathrm{i}(t) - \beta_\mathrm{V_i} (V^* - V_\mathrm{i}(t)), 
\end{align}
where $\beta_\mathrm{Q_i}$ and $\beta_\mathrm{V_i}$ are the free parameters that lead to different steady-state characteristics of the system as discussed later.
\begin{lemma}\label{lem:static_solution}
    Let graph $\mathcal{G}$ be strongly connected (i.e. $\mathcal{G}$ has a directed path between every pair of distinct nodes). Then, the solution $x_\mathrm{t}^*$ for problem~\eqref{eq:opt_prob} is 
    $x_\mathrm{t}^*:= \frac{1}{N(1+\gamma)} \sum_{i=1}^N \alpha_\mathrm{i}(t)$. 
\end{lemma}
\begin{proof}
    Let $x_\mathrm{i} = x_\mathrm{j} = x, \forall i,j$. The separable form of the objective leads to an equivalent problem, with $\alpha_\mathrm{i}(t)$ as in  \eqref{eq:alpha_constraint},
    \begin{align}\label{eq:opt_prob_equivalent}
    \minimize_{x} &  \textstyle \  \sum_{i=1}^N  \frac{1}{2}(x - \alpha_\mathrm{i}(t))^2 + \frac{\gamma}{2} x^2.
    \end{align}
As problem~\eqref{eq:opt_prob_equivalent} is unconstrained we can differentiate the objective and set it to zero to determine the optimal solution. In particular, $\sum_{i=1}^N [ (1+\gamma) x_\mathrm{t}^* - \alpha_\mathrm{i}(t)] = 0$ which implies. $x_\mathrm{t}^* = \frac{1}{N(1+\gamma)} \sum_{i=1}^N \alpha_\mathrm{i}(t)$. Therefore, $x_\mathrm{i} = x_\mathrm{j} = x_\mathrm{t}^*$. 
\end{proof}
\begin{remark}
    Note that in problem~\eqref{eq:opt_prob} and the update~\eqref{eq:adjustmet_update} any GFM IBR $i$ only requires local measurements. Further, the GFM IBRs can choose the free parameters based on their local requirements without centralized constraints.
\end{remark}
\noindent Delaying the discussion on how problem~\eqref{eq:opt_prob} is solved, we first present the steady state analysis for the voltage and reactive power dynamics of the GFM IBRs. Assuming the overall system consisting of the loads and the IBRs has a steady state, let $t_\mathrm{ss}$ denote the time after which the system achieves a steady state. Therefore, for any time $t > t_\mathrm{ss}$, $V_\mathrm{i}(t) = V_\mathrm{i}(t_\mathrm{ss}), Q_\mathrm{i}(t) = Q_\mathrm{i}(t_\mathrm{ss})$. Then from~\eqref{eq:voltage_droop} and~\eqref{eq:adjustmet_update}, 
\begin{align*} 
& V_\mathrm{i}(t_\mathrm{ss}) = V^* - m_\mathrm{i}Q_\mathrm{i}(t_\mathrm{ss}) + v_\mathrm{i}(t_\mathrm{ss}) \\ 
& \hspace{-0.05in} =  V^* - m_\mathrm{i}Q_\mathrm{i}(t_\mathrm{ss}) + x_{\mathrm{t}_\mathrm{ss}}^* + \beta_\mathrm{Q_i} m_\mathrm{i} Q_\mathrm{i}(t_\mathrm{ss}) - \beta_\mathrm{V_i} (V^* - V_\mathrm{i}(t_\mathrm{ss})) \\
& \hspace{-0.05in} =  (1 - \beta_\mathrm{V_i}) V^* - ( 1 - \beta_\mathrm{Q_i}) m_\mathrm{i}Q_\mathrm{i}(t_\mathrm{ss}) + x_\mathrm{t_{ss}}^* + \beta_\mathrm{V_i} V_\mathrm{i}(t_\mathrm{ss}). 
\end{align*}
\begin{align} \label{eq:steady_state}
&\hspace{-0.02in} \mbox{Thus,} \ (1 - \beta_\mathrm{V_i}) V_\mathrm{i}(t_\mathrm{ss}) + ( 1 - \beta_\mathrm{Q_i}) m_\mathrm{i}Q_\mathrm{i}(t_\mathrm{ss}) = (1 - \beta_\mathrm{V_i}) V^*  \nonumber \\
& + x_\mathrm{t_{ss}}^* = (1 - \beta_\mathrm{V_i}) V^* + \textstyle \frac{1}{N(1+\gamma)} \big[ \sum_{l = 1}^N a_\mathrm{V_l}(V^* - V_\mathrm{l}(t_\mathrm{ss})) \nonumber \\
& \hspace{1.7in} \textstyle + 
\sum_{l = 1}^N a_\mathrm{Q_l} m_\mathrm{l} Q_\mathrm{l}(t_\mathrm{ss}) \big]. 
\end{align}
Note that different choices of the free design parameters achieve different objectives in the steady state. In particular,
\begin{itemize}
    \item [(i)] let $\beta_\mathrm{V_i} = \beta_\mathrm{V} = 1, \beta_\mathrm{Q_i} = \beta_\mathrm{Q} = 0$, for all $i \in \{1,2,\dots, N\}$. Then from~\eqref{eq:steady_state},  $m_\mathrm{i} Q_\mathrm{i}(t_\mathrm{ss}) = m_\mathrm{j} Q_\mathrm{j}(t_\mathrm{ss})$ for all $i,j \in \{1,2,\dots, N\}$. Thus, equal-rated reactive power sharing is achieved. 
    \item [(ii)] let $\beta_\mathrm{V_i} = \beta_\mathrm{V} = 0, \beta_\mathrm{Q_i} = \beta_\mathrm{Q} = 1, a_\mathrm{V_i} = a_\mathrm{Q_i} = 0$, for all $i \in \{1,2,\dots, N\}$. Then from~\eqref{eq:steady_state}, $V_\mathrm{i}(t_\mathrm{ss}) = V^*$ for all $i \in \{1,2,\dots, N\}$. Thus, voltage regulation to $V^*$ is achieved.
    \item [(iii)] a combination of the above cases can also be achieved, i.e. let $\beta_\mathrm{V_i} = \beta_\mathrm{V} = 1, \beta_\mathrm{Q_i} = \beta_\mathrm{Q} = 0$ for certain buses towards the objective of the equal-rated power reactive sharing, while, other buses have $\beta_\mathrm{V_i} = \beta_\mathrm{V} = 0, \beta_\mathrm{Q_i} = \beta_\mathrm{Q} = 1, a_\mathrm{V_i} = a_\mathrm{Q_i} = 0$ to maintain voltage regulation. 
\end{itemize}
\begin{remark}
    Problem~\eqref{eq:opt_prob} and update~\eqref{eq:adjustmet_update} 
    provide a versatile framework to allow the individual GFM IBRs to strike diverse trade-offs wherein, for example, a set of GFM IBRs might want to emphasize voltage regulation while others emphasize how reactive power is shared. This is the first control scheme with such a capability. We demonstrate this powerful ability of our framework via CHIL experiments in Section~\ref{sec:results}. 
\end{remark}
Next, we discuss how the solution to problem~\eqref{eq:opt_prob} is obtained. We first review the \textit{GradConsensus} method developed in \cite{GradConsensus} and then utilize it to present the algorithm to generate the update law~\eqref{eq:adjustmet_update} at any time $t$. 
\subsection{\textit{GradConsensus} Method}\label{sec:GradCons_review}
The \textit{GradConsensus} method developed in \cite{GradConsensus} solves a static distributed optimization problem of the form 
\begin{align}\label{eq:prob_cons_constraints}
  & \hspace{0.6in} \minimize\limits_{x_1,\dots,x_\mathrm{n}} \textstyle \ \sum_{i=1}^{n} f_\mathrm{i}(x_\mathrm{i})\\
  & \mbox{subject to} \ x_\mathrm{i} = \ x_\mathrm{j}, \  \mbox{for all} \ i,j \in \{1,2,\dots,N\}, \nonumber
\end{align}
where $x_\mathrm{i}$ is an estimate maintained locally at agent $i$. The local estimates of node $i$ at any iteration $k$ are updated as,
\begin{align}
    x_\mathrm{i}^{+} &:= x_\mathrm{i}(k-1) - \rho \nabla f_\mathrm{i}(x_\mathrm{i}(k-1)) \label{eq:gradcons_update1} \\
    x_\mathrm{i}(k) &= \mbox{Consensus}_{\varepsilon(k-1)}(x_\mathrm{i}^{+}), \label{eq:gradcons_update2}
\end{align}
where $\nabla f_\mathrm{i}(x_\mathrm{i}(k-1))$ is the gradient of the function $f_\mathrm{i}$ evaluated at $x_\mathrm{i}(k-1)$ and Consensus$_{\varepsilon}(y^0)$ is a distributed approximate consensus protocol initialized with the value $y^0 \in \mathbb{R}^N$ and a precision tolerance $\varepsilon$. Next, we describe the Consensus$_{\varepsilon}$ protocol. Under the Consensus$_{\varepsilon}$ protocol every node $i$ maintains estimates $y_\mathrm{i}, z_\mathrm{i}, r_\mathrm{i}, M_\mathrm{i}, m_\mathrm{i}$. The estimates $ y_\mathrm{i}, z_\mathrm{i}, r_\mathrm{i}$ have the following local distributed updates
\begin{align}
  y_{\mathrm{i}}(k+1) &= \textstyle \frac{1}{\mathcal{N}_\mathrm{i}^{+} + 1}y_{\mathrm{i}}(k)+\sum_{ j \in \mathcal{N}_\mathrm{i}^{-} } \hat{y}_{\mathrm{j}}(k)\label{eq:numerator}\\
  z_{\mathrm{i}}(k+1) &= \textstyle \frac{1}{\mathcal{N}_\mathrm{i}^{+} + 1} z_{\mathrm{i}}(k)+\sum_{ j \in \mathcal{N}_\mathrm{i}^{-} } \hat{z}_{\mathrm{j}}(k)\label{eq:denominator} \\
  r_\mathrm{i}(k+1) &= \textstyle \frac{1}{z_\mathrm{i}(k+1)}y_\mathrm{i}(k+1), \label{eq:ratio}
\end{align}
with $y_\mathrm{i}(0) = y^0_\mathrm{i}, z_\mathrm{i}(0) = 1$ for all $i\in \mathcal{E}$. Here, $\hat{y}_{\mathrm{j}}(k),\hat{z}_{\mathrm{j}}(k)$ are the estimates sent by neighboring node $j$ to agent $i$ with  $\hat{y}_{\mathrm{j}}(k) = \frac{1}{\mathcal{N}_\mathrm{j}^{+} + 1} y_\mathrm{j}(k), \hat{z}_{\mathrm{j}}(k) = \frac{1}{\mathcal{N}_\mathrm{j}^{+} + 1} z_\mathrm{j}(k)$.
The estimates $M_\mathrm{i}, m_\mathrm{i}$ are updated as: Let $ \mathcal{N}_\mathrm{i}^u = \mathcal{N}_\mathrm{i}^{-} \cup \{\mathrm{i}\}$,
\begin{align}
    M_\mathrm{i}(k+1) &= \max_{j \in \mathcal{N}_\mathrm{i}^u} M_\mathrm{j}(k), \ m_\mathrm{i}(k+1) = \min_{j \in \mathcal{N}_\mathrm{i}^u } m_\mathrm{j}(k), \label{eq:max_min_update}
\end{align}
with $M_\mathrm{i}(0)  = m_\mathrm{i}(0) = r_\mathrm{i}(0)$, for all $i \in \mathcal{E}$. Given a tolerance $\varepsilon$, iterations~\eqref{eq:numerator}-\eqref{eq:ratio} are terminated after $k_\varepsilon$ such that $|M_\mathrm{i}(k) - m_\mathrm{i}(k)| \leq \varepsilon$, for $k > k_\varepsilon$ \cite{GradConsensus}. It is established in \cite{GradConsensus} that at any iteration $k_\mathrm{end}$ when the \textit{GradConsensus} algorithm is terminated the estimates $r_\mathrm{i}(k_\mathrm{end})$ and $r_\mathrm{j}(k_\mathrm{end})$ of any two agents $i$ and $j$ are such that $|r_\mathrm{i}(k_\mathrm{end}) - r_\mathrm{j}(k_\mathrm{end})| \leq \varepsilon$ and $|r_\mathrm{i}(k_\mathrm{end}) - \frac{1}{N} \sum_{j=1}^N y^0_\mathrm{j}| \leq \varepsilon$ for all $i,j \in \{1,2,\dots,N\}.$
\begin{remark}
    The iterations~\eqref{eq:numerator}-\eqref{eq:max_min_update} are amenable to distributed synthesis (\cite{ADMM_tac}) as a GFM IBR only requires the communicated estimates $\hat{y}_\mathrm{j}, \hat{z}_\mathrm{j}$ from its neighbors for updating its states, no additional coordination is required. Further, as any GFM IBR $j$ joins or leaves the network, GFM IBR $i$ can update its neighborhood sets $\mathcal{N}_\mathrm{i}^{-},\mathcal{N}_\mathrm{i}^{+}$ locally not requiring a centralized re-design. Hence, the GFM IBRs under the proposed algorithm operate in a plug-and-play manner. 
\end{remark}
\subsection{Solution to Problem~\eqref{eq:opt_prob}}\label{sec:opt_prob_sol}
Problem~\eqref{eq:opt_prob} is instantiated at every time $t$. The constraint $x_\mathrm{i} = x_\mathrm{j}$ for all $i,j \in \{1,2,\dots, N\}$ couples the local estimates of all the GFM IBRs. Thus, a \textit{distributed} strategy that allows the GFM IBRs to obtain an estimate of the solution to problem~\eqref{eq:opt_prob} while satisfying the coupling constraint (within a prescribed tolerance) is needed. We adopt the iterations~\eqref{eq:gradcons_update1} and~\eqref{eq:gradcons_update2} to provide such a distributed solution. Let $x_\mathrm{i}(t)$ denote the estimate maintained at GFM IBR $i$ at time $t$ and $\varepsilon(t) > 0$ be any specified precision tolerance. Then the estimate $x_\mathrm{i}(t)$ is updated as
\begin{align}
    x_\mathrm{i}^{+} :&= x_\mathrm{i}(t-1) + \rho( (1+\gamma) x_\mathrm{i}(t-1) - \alpha_\mathrm{i}(t-1)) , \label{eq:gradcons_opt_update1} \\
    x_\mathrm{i}(t) &= \mbox{Consensus}_{\varepsilon(t-1)}(x_\mathrm{i}^{+}), \label{eq:gradcons_opt_update2}
\end{align}
At any time $t$, $x_\mathrm{i}(t)$ is the estimate of the optimal solution $x_t^*$ with GFM IBR $i$. Thus, having $x_\mathrm{i}(t)$ at hand any GFM IBR $i$ can determine the adjustment $v_\mathrm{i}(t)$ at time $t$ as, 
\begin{align}\label{eq:adjustment_update_online}
    v_\mathrm{i}(t) := x_\mathrm{i}(t) + \beta_\mathrm{Q_i} m_\mathrm{i} Q_\mathrm{i}(t) - \beta_\mathrm{V_i} (V^* - V_\mathrm{i}(t)). 
\end{align}
Next, we present a convergence result for updates~\eqref{eq:gradcons_opt_update1}-\eqref{eq:adjustment_update_online}. 
\begin{theorem}\label{thm:convproof}
Let graph $\mathcal{G}$ be strongly connected (i.e. $\mathcal{G}$ has a directed path between every pair of distinct nodes). Assume, that the iterations~\eqref{eq:gradcons_opt_update1}-\eqref{eq:adjustment_update_online} achieve a steady state after time $t_\mathrm{ss}$. Given, $\varepsilon >0$, $| x_\mathrm{i}(t) - x^*_\mathrm{t} | \leq \varepsilon$, for $t \geq t_\mathrm{ss}, \forall i$.
\end{theorem}
\begin{proof}
    From~\eqref{eq:gradcons_opt_update1} and~\eqref{eq:gradcons_opt_update2}, at any time $t$ and any $i$,
    \begin{align}\label{eq:dummy1}
       x_\mathrm{i}(t+1) = \textstyle \frac{(1 + \rho(1+\gamma))}{N}\sum_{j=1}^N x_\mathrm{j}(t) - \frac{\rho}{N} \sum_{i=1}^N\alpha_\mathrm{i}(t) + \xi_\mathrm{i}(t),
    \end{align}
    where $|\xi_\mathrm{i}(t)| \leq \varepsilon(t)$. Taking the average over node index $i$,
    \begin{align*}
       \textstyle \frac{1}{N} \sum_{i=1}^N x_\mathrm{i}(t+1) & = \textstyle \frac{(1 + \rho(1+\gamma))}{N}\sum_{j=1}^N x_\mathrm{j}(t) \\
        & \hspace{0.3in} \textstyle - \frac{\rho}{N} \sum_{i=1}^N\alpha_\mathrm{i}(t) + \frac{1}{N} \sum_{i=1}^N \xi_\mathrm{i}(t).
    \end{align*}
    For $t > t_\mathrm{ss}, \frac{1}{N} \sum_{i=1}^N x_\mathrm{i}(t_\mathrm{ss}) = \frac{(1 + \rho(1+\gamma))}{N}\sum_{j=1}^N x_\mathrm{j}(t_\mathrm{ss}) - \frac{\rho}{N} \sum_{i=1}^N\alpha_\mathrm{i}(t) + \frac{1}{N} \sum_{i=1}^N \xi_\mathrm{i}(t).$ 
    
    \noindent Therefore, from~\eqref{eq:dummy1}, $\frac{ \rho(1+\gamma)}{N}\sum_{j=1}^N x_\mathrm{j}(t_\mathrm{ss})- \textstyle \frac{\rho}{N} \sum_{i=1}^N\alpha_\mathrm{i}(t) = x_\mathrm{i}(t_\mathrm{ss}) - \frac{1}{N}\sum_{j=1}^N x_\mathrm{j}(t_\mathrm{ss}) - \xi_t(t_\mathrm{ss}) =  - \frac{1}{N} \sum_{i=1}^N \xi_\mathrm{i}(t).$ Thus,
    \begin{align*}
        \textstyle x_\mathrm{i}(t_\mathrm{ss}) &= \textstyle \frac{1}{N}\sum_{j=1}^N x_\mathrm{j}(t_\mathrm{ss}) + \xi_\mathrm{i}(t) - \frac{1}{N} \sum_{i=1}^N \xi_\mathrm{i}(t)\\
        & = \textstyle \frac{\rho (\sum_{i=1}^N\alpha_\mathrm{i}(t) )}{N (\rho(1+\gamma))} - \frac{\sum_{i=1}^N \xi_\mathrm{i}(t) }{N(\rho(1+\gamma))} + \xi_\mathrm{i}(t) - \frac{\sum_{i=1}^N \xi_\mathrm{i}(t)}{N} 
    \end{align*}
    Let $\varepsilon(t) < \varepsilon/2$ for $t > t_\mathrm{ss}$, using the result of Lemma~\ref{lem:static_solution},
    \begin{align*}
        & \textstyle | x_\mathrm{i}(t_\mathrm{ss}) - x_\mathrm{t}^* | =  \left | x_\mathrm{i}(t_\mathrm{ss}) - \frac{\sum_{i=1}^N\alpha_\mathrm{i}(t) }{N (1+\gamma)} \right | \\
        & \textstyle  \leq \left | \frac{1-\rho(1+\gamma)}{N\rho(1+\gamma)} \sum_{i=1}^N \xi_\mathrm{i}(t) + \xi_\mathrm{i}(t) \right| \leq \frac{1}{N} \sum_{i=1}^N |\xi_\mathrm{i}(t)| + |\xi_\mathrm{i}(t)| 
    \end{align*}  
    \hspace{0.03in} $\leq \textstyle \frac{1}{N} \sum_{i=1}^N \varepsilon(t) + \varepsilon(t) \leq \varepsilon.$  
\end{proof}
\begin{remark}\label{rem:imperfect_scenarios}
During updates~\eqref{eq:gradcons_opt_update1}-\eqref{eq:adjustment_update_online} the GFM IBRs share a non-linear estimate of their measurements, obscuring the private information and hence preventing malicious cyber-attacks. Moreover, the iterations~\eqref{eq:numerator}-\eqref{eq:max_min_update} can be easily modified to work under various practical scenarios incorporating time-delays, dynamic GFM IBR topology \cite{switching_top}, presence of communication noise \cite{nrpushsum}. Further, the detection of the malicious behaving GFM IBRs can also be achieved \cite{intruder_codit}. 
\end{remark}
\renewcommand{\arraystretch}{1.2}
\begin{table}[b]
\centering
\caption{$1$-PHASE GFM IBR UNDER STUDY}
\label{table:data}
\scriptsize
\begin{tabular}{c|c}
\hline 
$\mathbf{Inverter}$ & $\mathbf{Value}$    \\ \hline 
$\mathrm{Ratings}$ ($1$-$\phi$) & $240~\mathrm{V}$, $60~\mathrm{Hz}$, $200~\mathrm{kVA}$ \\ \hline
$\mathrm{Inverter~Parameters}$ & $V_\mathrm{dc}$ = $500~\mathrm{V}$, $f_\mathrm{sw}$ = $20~\mathrm{kHz}$ \\ \hline
$\mathrm{Filter~Parameters}$ & $L_\mathrm{f}$ = $0.02~\mathrm{mH}$, $R_\mathrm{f}$ = $4$~$\mathrm{m}\Omega$, $C_\mathrm{f}$ = $20~\mu\mathrm{F}$ \\ \hline \hline
$\mathbf{Control}$ & $\mathbf{Value}$ \\ \hline 
$\mathrm{Droop~Parameters}$ & $n_i$ = $0.32$~$\mathrm{Hz/kW}$, $m_i$ = $40$~$\mathrm{V/kVAr}$ \\ \hline
\end{tabular}
\end{table}

\section{CHIL Demonstration and Results}\label{sec:results}
In this section, we will evaluate the performance of our approach using a CHIL-based demonstration. 
\subsection{Experimental Configuration}
\begin{figure}[b]
    \centering
    \includegraphics[scale=0.25,trim={0cm 0cm 3.5cm 0cm},clip]{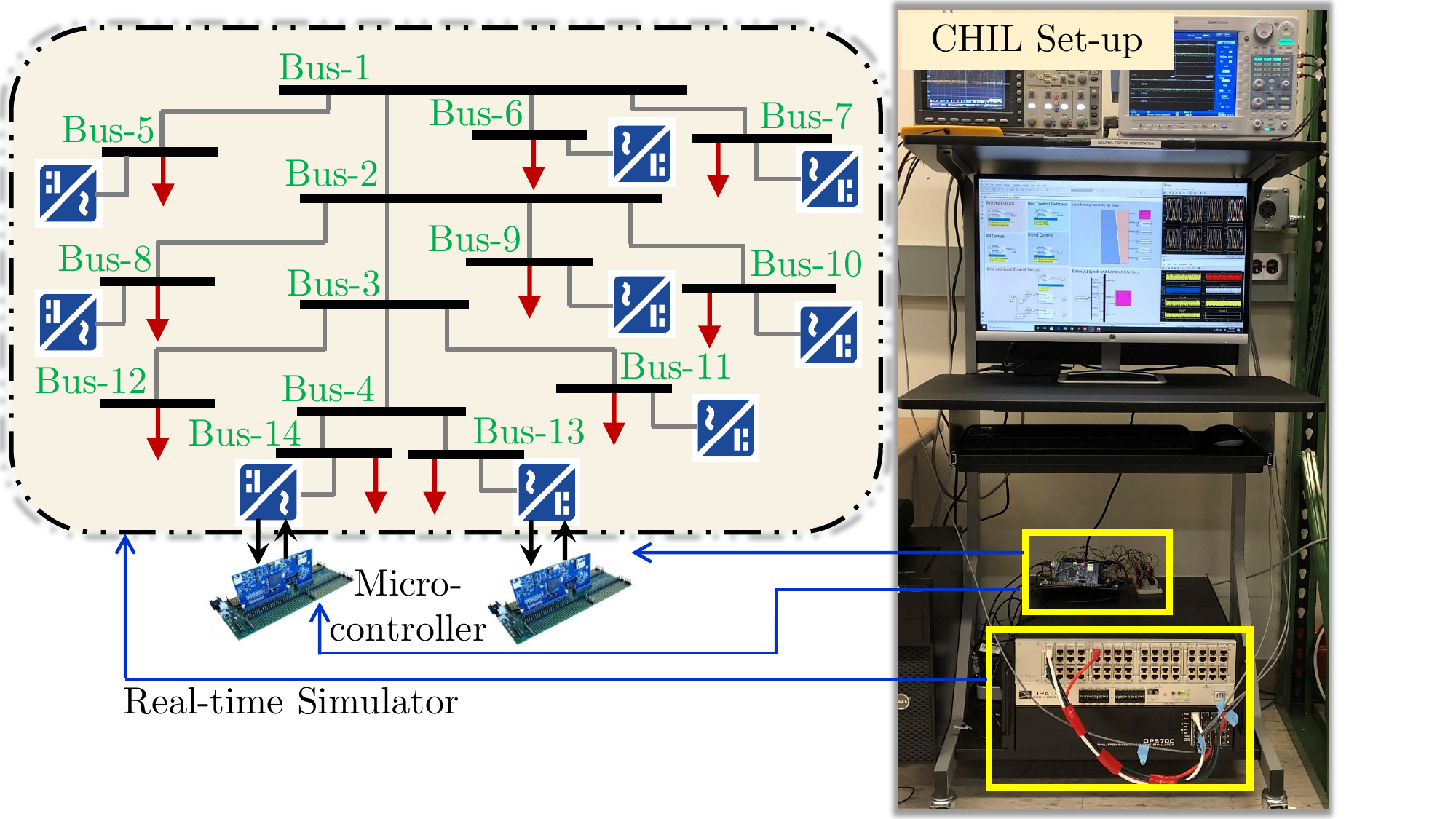}
    \caption{The laboratory-based system-in-the-loop (SIL) and controller hardware-in-the-loop (CHIL) experimental hardware setup.}
    \label{fig:test_system}
\end{figure}
We emulate the residential-scale sub-network of the North American low-voltage distribution feeder (single phase, $240$V, $60$Hz, $1.25$MVA) from CIGRE Task Force \cite{strunz2009developing} (see Fig.~\ref{fig:test_system}). It has $10$ GFM IBRs connected at various buses in the network, resulting in an asymmetrical connection. Fig.~\ref{fig:test_system} provides the location of the GFM IBRs. In this demonstration, we assume that each GFM IBR has equal ratings. The control logic for the inner-voltage loop and droop-based primary layer are adopted from \cite{sohamGFM1,sohamGFM2}. GFM IBRs are modeled with single-phase H-bridge topology with $\mathrm{LC}$ filters ($L_\mathrm{f}$, $R_\mathrm{f}$, and $C_\mathrm{f}$ as inverter side filter inductors, parasitic resistances and filter capacitor respectively) with electrical parameters,m tabulated in Table~\ref{table:data}. The droop control parameters of the GFM IBRs are tabulated in Table~\ref{table:data}. On top of the power layer, there is a communication layer that enables the transmission of various measurements, from the North American low-voltage distribution feeder to the optimization engine, and various control signals, from the optimization engine to the North American low-voltage distribution feeder. There are three important components in the demonstration: i) the North American low-voltage distribution feeder and the GFM IBRs, emulated using the eMEGASIM platform with a simulation step size of $50 \mu$s inside the OP$5700$ RT-simulator (RTS), which is manufactured by OPAL-RT and interfaced with two low-cost Texas Instruments TMS$320$F$28379$D, $16$/$12$-bit floating-point $200$-MHz Delfino microcontroller boards for controlling two of the GFM IBRs; ii) the optimization engine of the IBRs running $10$ independent instances via $10$ parallel processes running the optimization iterations in Python $3.7.1$ on a laptop with $16$ GB RAM and an Intel Core i$7$ processor running at $1.90$ GHz. The GFM IBRs are connected via a directed communication network having a diameter of $5$ nodes; and iii) a Device Control Gateway (DCG), realized using standard User Datagram Protocol (UDP) \cite{udp}, that interfaces the emulated system and the optimization engine and forms the communication layer that continuously listens to the measurements sent from the power network and adjustment commands from the optimization engine. 

\subsection{Results and Discussions}
Three test cases are demonstrated by emulating a sequence of events as follows:\\
$\bullet~\mathtt{CASE}$-$\mathtt{1}$: GFM IBRs operated without secondary control compensation until about $4.5$ s. Then, at $4.5$ s, distributed secondary control began, evenly sharing reactive power among GFM IBRs $1$ to $10$. At $13.5$ s, system demand rose from $1000$ kW to $1250$ kW for active power and from $750$ kVAr to $1125$ kVAr for reactive power. At $24$ s, demand decreased back to $1000$ kW for active power and $750$ kVAr for reactive power.\\
$\bullet~\mathtt{CASE}$-$\mathtt{2}$: Until approximately $23$ seconds, all GFM IBRs operated without the proposed secondary control compensation. Then, at $23$ seconds, distributed secondary control for voltage regulation among GFM IBRs $1$ to $10$ began. By $34$ seconds, the total system demand increased to $1250$ kW for active power and $1125$ kVAr for reactive power from $1000$ kW and $750$ kVAr respectively. Subsequently, at $47$ seconds, the total system demand decreased back to $1000$ kW for active power and $750$ kVAr for reactive power.\\
$\bullet~\mathtt{CASE}$-$\mathtt{3}$: Until around $6$ seconds, all GFM IBRs operated without the proposed secondary control compensation. Then, at $6$ seconds, distributed secondary control began, with GFM IBRs $1$ to $7$ equally sharing reactive power, while GFM IBRs $8$ to $10$ regulated voltage. By $11.5$ seconds, the total system demand increased to $1250$ kW for active power and $1125$ kVAr for reactive power from $1000$ kW and $750$ kVAr respectively. Subsequently, at $17$ seconds, the total system demand decreased back to $1000$ kW for active power and $750$ kVAr for reactive power.
\par Figs.~\ref{fig:combined_results}(i)-(iv) provide the results of $\mathtt{CASE}$-$\mathtt{1}$. Here, the GFM IBRs share the total reactive power demand by providing an equal amount of $75$ kVAr ($0.375$ pu) during $4.5$ to $13.5$ seconds during the reactive power sharing objective. The developed scheme tracks both the increase and the decrease in load with all the GFM IBRs sharing reactive power equally as shown in Fig.~\ref{fig:combined_results}(ii). The voltages at the GFM IBRs however are not tightly regulated as seen in Figs.~\ref{fig:combined_results}(iii). Fig.~\ref{fig:combined_results}(i) and Fig.~\ref{fig:combined_results}(iv) show that the proposed secondary control is not influencing the active power sharing functionality and as result the frequency of the system is also kept almost unchanged with minimal transients. 
\par Figs.~\ref{fig:combined_results}(v)-(viii) provide the results of $\mathtt{CASE}$-$\mathtt{2}$. It can be seen from Fig.~\ref{fig:combined_results}(vii) that the voltages of the GFM IBRs are tightly regulated around the reference voltage under $\mathtt{CASE}$-$\mathtt{2}$. The total reactive power is not shared equally among the GFM IBRs as seen in Fig.~\ref{fig:combined_results}(vi). Similarly here also, as shown in Fig.~\ref{fig:combined_results}(v) and Fig.~\ref{fig:combined_results}(viii), the proposed secondary control is not influencing the active power sharing functionality and the system frequency has minimal transients.
\par Figs.~\ref{fig:combined_results}(ix)-(xii) present the system state under $\mathtt{CASE}$-$\mathtt{3}$. Here, three GFM IBRs (GFMs $8-10$) have unequal sharing of reactive power while all the remaining seven GFM IBRs (GFMs $1-7$) have their voltages tightly regulated. This shows the versatility of our proposed scheme in achieving both voltage regulation and equal-rated reactive power-sharing objectives. Similarly, the control doesn't not adversely affect the active power sharing of the GFM IBRs and the frequency of the system is maintained around the nominal value of $60$ Hz as seen in Figs.~\ref{fig:combined_results}(ix) and Figs.~\ref{fig:combined_results}(xii).
\begin{figure*}[t]
    \centering
    \includegraphics[scale=0.85,trim={0.1cm 0cm 0.1cm 0.1cm},clip]{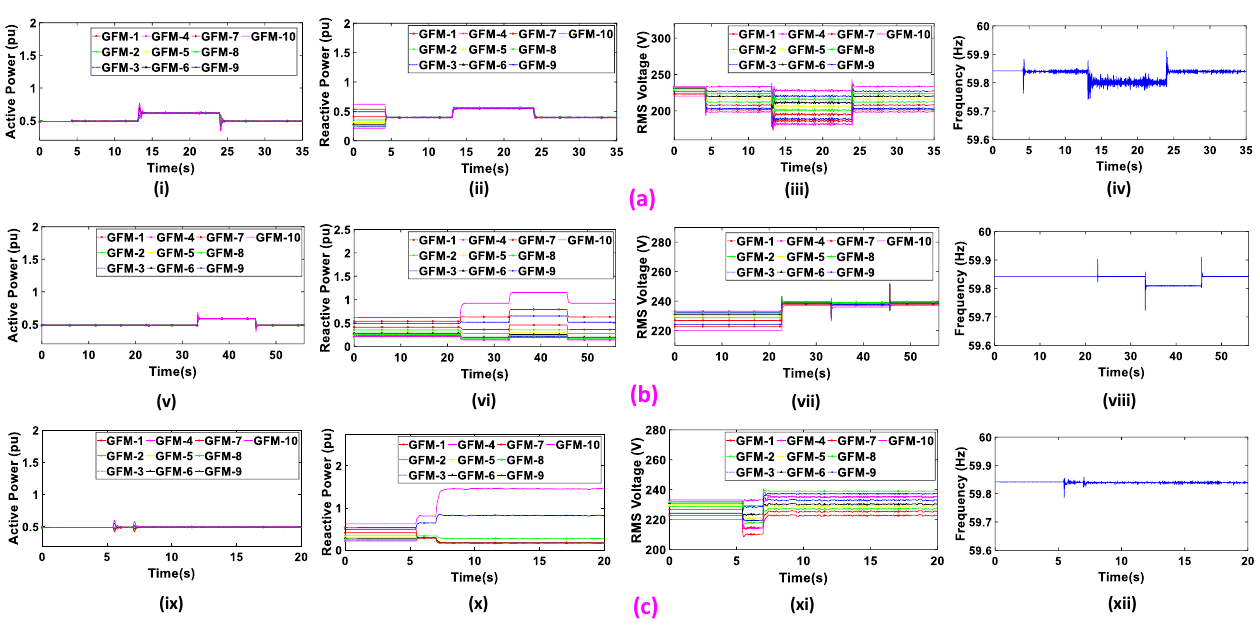}
    \caption{Performance results of the developed distributed secondary controller. (a) Under $\mathtt{CASE}$-$1$: equal-rated reactive power sharing, (b) Under $\mathtt{CASE}$-$2$: tight voltage regulation, (c) Under $\mathtt{CASE}$-$3$: mixed objective DERs $1$-$7$ equal-rated reactive power sharing and DERs $8$-$10$ tight voltage regulation.}
    \label{fig:combined_results}
\end{figure*}

\section{Conclusion and Future Work}\label{sec:conclusion}
In this article, we considered a secondary control problem in MGs with heterogeneous GFM IBRs connected via varying line impedances. In contrast to the literature, we proposed a fully distributed controller scheme that can be both synthesized and implemented at each IBR using only locally available information. The salient features of the proposed scheme are its plug-and-play operation, versatility to achieve different control objectives, and privacy-preserving updates preventing the dissemination of private information between different GFM IBRs. We presented a convergence analysis establishing the voltage regulation and reactive power sharing properties. The CHIL experiment results corroborate the efficacy of the proposed distributed controller in secondary control scenarios.
We remark that the developed distributed controller can be extended to the scenarios of communication delays and dynamic interconnection between the GFM DERs and variations of line impedances between the DERs. However, an exact characterization is beyond the scope of the current article and is pursued as a future direction of research.

\bibliography{references}

\begin{thebibliography}{10}
\providecommand{\url}[1]{#1}
\csname url@samestyle\endcsname
\providecommand{\newblock}{\relax}
\providecommand{\bibinfo}[2]{#2}
\providecommand{\BIBentrySTDinterwordspacing}{\spaceskip=0pt\relax}
\providecommand{\BIBentryALTinterwordstretchfactor}{4}
\providecommand{\BIBentryALTinterwordspacing}{\spaceskip=\fontdimen2\font plus
\BIBentryALTinterwordstretchfactor\fontdimen3\font minus
  \fontdimen4\font\relax}
\providecommand{\BIBforeignlanguage}[2]{{%
\expandafter\ifx\csname l@#1\endcsname\relax
\typeout{** WARNING: IEEEtran.bst: No hyphenation pattern has been}%
\typeout{** loaded for the language `#1'. Using the pattern for}%
\typeout{** the default language instead.}%
\else
\language=\csname l@#1\endcsname
\fi
#2}}
\providecommand{\BIBdecl}{\relax}
\BIBdecl

\bibitem{guerrero2010hierarchical}
J.~M. Guerrero, J.~C. Vasquez, J.~Matas, L.~G. De~Vicu{\~n}a, and M.~Castilla,
  ``Hierarchical control of droop-controlled ac and dc microgrids—a general
  approach toward standardization,'' \emph{IEEE Transactions on industrial
  electronics}, vol.~58, no.~1, pp. 158--172, 2010.

\bibitem{chandorkar1993control}
M.~C. Chandorkar, D.~M. Divan, and R.~Adapa, ``Control of parallel connected
  inverters in standalone ac supply systems,'' \emph{IEEE transactions on
  industry applications}, vol.~29, no.~1, pp. 136--143, 1993.

\bibitem{zhong2011robust}
Q.-C. Zhong, ``Robust droop controller for accurate proportional load sharing
  among inverters operated in parallel,'' \emph{IEEE Transactions on industrial
  Electronics}, vol.~60, no.~4, pp. 1281--1290, 2011.

\bibitem{savaghebi2012secondary}
M.~Savaghebi, A.~Jalilian, J.~C. Vasquez, and J.~M. Guerrero, ``Secondary
  control scheme for voltage unbalance compensation in an islanded
  droop-controlled microgrid,'' \emph{IEEE Transactions on Smart Grid}, vol.~3,
  no.~2, pp. 797--807, 2012.

\bibitem{micallef2014reactive}
A.~Micallef, M.~Apap, C.~Spiteri-Staines, J.~M. Guerrero, and J.~C. Vasquez,
  ``Reactive power sharing and voltage harmonic distortion compensation of
  droop controlled single phase islanded microgrids,'' \emph{IEEE Transactions
  on Smart Grid}, vol.~5, no.~3, pp. 1149--1158, 2014.

\bibitem{shafiee2012distributed}
Q.~Shafiee, J.~C. Vasquez, and J.~M. Guerrero, ``Distributed secondary control
  for islanded microgrids-a networked control systems approach,'' in
  \emph{IECON 2012-38th Annual Conference on IEEE Industrial Electronics
  Society}.\hskip 1em plus 0.5em minus 0.4em\relax IEEE, 2012, pp. 5637--5642.

\bibitem{shafiee2013distributed}
Q.~Shafiee, J.~M. Guerrero, and J.~C. Vasquez, ``Distributed secondary control
  for islanded microgrids—a novel approach,'' \emph{IEEE Transactions on
  power electronics}, vol.~29, no.~2, pp. 1018--1031, 2013.

\bibitem{shafiee2013robust}
Q.~Shafiee, {\v{C}}.~Stefanovi{\'c}, T.~Dragi{\v{c}}evi{\'c}, P.~Popovski,
  J.~C. Vasquez, and J.~M. Guerrero, ``Robust networked control scheme for
  distributed secondary control of islanded microgrids,'' \emph{IEEE
  Transactions on Industrial Electronics}, vol.~61, no.~10, pp. 5363--5374,
  2013.

\bibitem{bidram2013secondary}
A.~Bidram, A.~Davoudi, F.~L. Lewis, and Z.~Qu, ``Secondary control of
  microgrids based on distributed cooperative control of multi-agent systems,''
  \emph{IET Generation, Transmission \& Distribution}, vol.~7, no.~8, pp.
  822--831, 2013.

\bibitem{dehkordi2016fully}
N.~M. Dehkordi, N.~Sadati, and M.~Hamzeh, ``Fully distributed cooperative
  secondary frequency and voltage control of islanded microgrids,'' \emph{IEEE
  Transactions on Energy Conversion}, vol.~32, no.~2, pp. 675--685, 2016.

\bibitem{simpson2013synchronization}
J.~W. Simpson-Porco, F.~D{\"o}rfler, and F.~Bullo, ``Synchronization and power
  sharing for droop-controlled inverters in islanded microgrids,''
  \emph{Automatica}, vol.~49, no.~9, pp. 2603--2611, 2013.

\bibitem{simpson2015secondary}
J.~W. Simpson-Porco, Q.~Shafiee, F.~D{\"o}rfler, J.~C. Vasquez, J.~M. Guerrero,
  and F.~Bullo, ``Secondary frequency and voltage control of islanded
  microgrids via distributed averaging,'' \emph{IEEE Transactions on Industrial
  Electronics}, vol.~62, no.~11, pp. 7025--7038, 2015.

\bibitem{chen2015distributed}
G.~Chen and E.~Feng, ``Distributed secondary control and optimal power sharing
  in microgrids,'' \emph{IEEE/CAA Journal of Automatica Sinica}, vol.~2, no.~3,
  pp. 304--312, 2015.

\bibitem{lu2016novel}
X.~Lu, X.~Yu, J.~Lai, Y.~Wang, and J.~M. Guerrero, ``A novel distributed
  secondary coordination control approach for islanded microgrids,'' \emph{IEEE
  Transactions on Smart Grid}, vol.~9, no.~4, pp. 2726--2740, 2016.

\bibitem{lu2016distributed}
X.~Lu, X.~Yu, J.~Lai, J.~M. Guerrero, and H.~Zhou, ``Distributed secondary
  voltage and frequency control for islanded microgrids with uncertain
  communication links,'' \emph{IEEE Transactions on Industrial Informatics},
  vol.~13, no.~2, pp. 448--460, 2016.

\bibitem{ding2018distributed}
L.~Ding, Q.-L. Han, and X.-M. Zhang, ``Distributed secondary control for active
  power sharing and frequency regulation in islanded microgrids using an
  event-triggered communication mechanism,'' \emph{IEEE Transactions on
  Industrial Informatics}, vol.~15, no.~7, pp. 3910--3922, 2018.

\bibitem{zhang2024secondary}
J.~Zhang, Y.~Men, L.~Ding, and X.~Lu, ``Secondary frequency and voltage
  regulation for inverter-based microgrids: A sparsity-promoting dapi control
  approach,'' \emph{IEEE Transactions on Control Systems Technology}, 2024.

\bibitem{zhu2015virtual}
Y.~Zhu, F.~Zhuo, F.~Wang, B.~Liu, R.~Gou, and Y.~Zhao, ``A virtual impedance
  optimization method for reactive power sharing in networked microgrid,''
  \emph{IEEE Transactions on Power Electronics}, vol.~31, no.~4, pp.
  2890--2904, 2015.

\bibitem{hu2018decentralised}
Y.~Hu, J.~Xiang, Y.~Peng, P.~Yang, and W.~Wei, ``Decentralised control for
  reactive power sharing using adaptive virtual impedance,'' \emph{IET
  Generation, Transmission \& Distribution}, vol.~12, no.~5, pp. 1198--1205,
  2018.

\bibitem{zhang2016distributed}
H.~Zhang, S.~Kim, Q.~Sun, and J.~Zhou, ``Distributed adaptive virtual impedance
  control for accurate reactive power sharing based on consensus control in
  microgrids,'' \emph{IEEE Transactions on Smart Grid}, vol.~8, no.~4, pp.
  1749--1761, 2016.

\bibitem{GradConsensus}
V.~Khatana, G.~Saraswat, S.~Patel, and M.~V. Salapaka, ``Gradconsensus:
  Linearly convergent algorithm for reducing disagreement in multi-agent
  optimization,'' \emph{IEEE Transactions on Network Science and Engineering},
  vol.~11, no.~1, pp. 1251--1264, 2024.

\bibitem{ADMM_tac}
V.~Khatana and M.~V. Salapaka, ``Dc-distadmm: Admm algorithm for constrained
  optimization over directed graphs,'' \emph{IEEE Transactions on Automatic
  Control}, vol.~68, no.~9, pp. 5365--5380, 2023.

\bibitem{switching_top}
G.~Saraswat, V.~Khatana, S.~Patel, and M.~V. Salapaka, ``Distributed
  finite-time termination for consensus algorithm in switching topologies,''
  \emph{IEEE Transactions on Network Science and Engineering}, vol.~10, no.~1,
  pp. 489--499, 2023.

\bibitem{nrpushsum}
V.~Khatana and M.~V. Salapaka, ``Noise resilient distributed average consensus
  over directed graphs,'' \emph{IEEE Transactions on Signal and Information
  Processing over Networks}, vol.~9, pp. 770--785, 2023.

\bibitem{intruder_codit}
S.~Patel, V.~Khatana, G.~Saraswat, and M.~V. Salapaka, ``Distributed detection
  of malicious attacks on consensus algorithms with applications in power
  networks,'' in \emph{2020 7th International Conference on Control, Decision
  and Information Technologies (CoDIT)}, vol.~1, 2020, pp. 397--402.

\bibitem{strunz2009developing}
K.~Strunz, R.~Fletcher, R.~Campbell, and F.~Gao, ``Developing benchmark models
  for low-voltage distribution feeders,'' in \emph{2009 IEEE Power \& Energy
  Society General Meeting}.\hskip 1em plus 0.5em minus 0.4em\relax IEEE, 2009,
  pp. 1--3.

\bibitem{sohamGFM1}
S.~Chakraborty, S.~Patel, and M.~V. Salapaka, ``$\mu$-synthesis-based
  generalized robust framework for grid-following and grid-forming inverters,''
  \emph{IEEE Transactions on Power Electronics}, vol.~38, no.~3, pp.
  3163--3179, 2023.

\bibitem{sohamGFM2}
------, ``Robust and optimal single-loop voltage controller for grid-forming
  voltage source inverters,'' in \emph{2020 IEEE Power and Energy Conference at
  Illinois (PECI)}, 2020, pp. 1--7.

\bibitem{udp}
\BIBentryALTinterwordspacing
``User datagram protocol,'' 2021, [November 18, 2021]. [Online]. Available:
  \url{https://en.wikipedia.org/wiki/User_Datagram_Protocol}
\BIBentrySTDinterwordspacing

\end{thebibliography}

\end{document}